\documentclass[a4paper]{llncs}
\providecommand{\url}[1]{#1}
\usepackage{amssymb,amsmath,enumerate,graphicx}

\usepackage{color}

\usepackage{tikz}
\usetikzlibrary{snakes}
\tikzstyle{hyb}=[rectangle,fill=green!50,draw,minimum size=3.5mm]
\tikzstyle{tre}=[circle,fill=green!50,draw,minimum size=3.5mm]
\newcommand{\etq}[1]{%
\draw (#1) node {$#1$};
}

\newcommand{\contained}{\preccurlyeq}

\newcommand{\pathgr}{\!\rightsquigarrow\!{}}

\renewcommand{\leq}{\leqslant}
\renewcommand{\geq}{\geqslant}

\begin{document}

\title{Two metrics for general phylogenetic networks}

\author{Gabriel Cardona\inst{1} \and Merc\`e Llabr\'es\inst{1} \and Francesc Rossell\'o\inst{1} \and
Gabriel Valiente\inst{2,3} } \authorrunning{G. Cardona et al.}
\institute{Department of Mathematics and Computer Science, University
of the Balearic Islands, E-07122 Palma de Mallorca,
\{\texttt{gabriel.cardona,merce.llabres,cesc.rossello}\}\texttt{@uib.es}
 \and
Algorithms, Bioinformatics, Complexity and Formal Methods Research
Group, Technical University of Catalonia, E-08034 Barcelona,
\texttt{valiente@lsi.upc.edu} }

\maketitle

\begin{abstract}
We prove that Nakhleh's latest   `metric' for phylogenetic networks separates distinguishable phylogenetic networks, and that a slight modification of it
provides a true distance on the class of all phylogenetic networks.
\end{abstract}
%

\section{Introduction}
\label{sec:intro}

L. Nakhleh has recently proposed a dissimilarity measure for the comparison of
phylogenetic networks \cite{nakhleh:2007}, but he has only proved that
it satisfies the separation axiom for metrics (zero distance means
isomorphism) on the class of all \emph{reduced} phylogenetic networks
in the sense of \cite{moret.ea:2004}.  And, although we show that this measure
separates phylogenetic networks more general than the reduced
ones (for instance, the tree-child phylogenetic networks \cite{cardona.ea:07b}), it does not satisfy the separation axiom on the whole class of
all phylogenetic networks (see Remark~\ref{rem2} below).

In this note we complement Nakhleh's work in two directions.  On the
one hand, we prove that, for this dissimilarity measure, zero distance implies
indistinguishability up to reduction in the sense of
\cite{moret.ea:2004}, a goal that had already been pursued by
Moret-Nakhleh-Warnow et al in \textsl{loc.~cit.}, failing in their attempt
\cite{cardona.ea:07a}.  In this way, and to the best of our
knowledge,  Nakleh's dissimilarity measure turns out to be is the first one  that separates distinguishable networks.
And, on the other hand, we show that a slight
modification of Nakhleh's definition does yield a true distance on the
whole class of all phylogenetic networks.  Again to the best of our
knowledge, this is the first true metric defined on this class.

\subsection{Notations}
\label{sec:prel}

Let $N=(V,E)$ be a DAG (a finite directed acyclic graph).  We say that a node
$v\in V$ is a \emph{child} of $u\in V$ if $(u,v)\in E$; we also say
then that $u$ is a \emph{parent} of $v$.  We say that a node is a
\emph{tree node} when it has at most one parent, and that it is a \emph{hybrid node} when it
has more than one parent.  A node that is not a leaf is called
\emph{internal}.  A DAG is \emph{rooted} when it has only one
\emph{root}: a node without parents.

A \emph{path} in $N$ is a sequence of nodes $(v_0,v_1,\dots,v_k)$ such
that $(v_{i-1},v_{i})\in E$ for all $i=1,\dots,k$.  We call $v_{0}$
the \emph{origin} of the path, $v_1,\ldots,v_{k-1}$ its \emph{intermediate nodes}, $v_{k}$ its \emph{end}, and $k$ its
\emph{length}.  We denote by $u\pathgr v$ any path with origin $u$ and
end $v$ and, whenever there exists a path $u\pathgr v$, we say that
$v$ is a \emph{descendant} of $u$.

The \emph{height} $h(v)$ of a node $v$ in a DAG $N$ is the largest
length of a path from $v$ to a leaf.  The absence of cycles implies
that the nodes of a DAG can be stratified by means of their heights:
the nodes of height 0 are the leaves, the nodes of height 1 are those
nodes all whose children are leaves, the nodes of height 2 are those
nodes all whose children are leaves and nodes of height 1, and so on.
If a node has height $m$, then all its children have height smaller
than $m$, and at least one of them has height exactly $m-1$.

Given a finite set $S$, an \emph{$S$-DAG} is 
a DAG whose leaves are bijectively labeled by elements of
$S$.  We shall always identify, usually without any further notice,
each leaf of an $S$-DAG with its label.  Two $S$-DAGs $N,N'$   are \emph{isomorphic}, in
symbols $N\cong N'$, when they are isomorphic as directed graphs and
the isomorphism preserves the leaves' labels.

 A \emph{phylogenetic network} on a set $S$ of taxa is a rooted $S$-DAG.

For every node $u$ of a phylogenetic network $N=(V,E)$, let $C(u)$ be
the set of all its descendants in $N$ and $N(u)$ the subgraph of $N$
supported on $C(u)$: it is still a phylogenetic network, with root $u$
and leaves labeled in the subset $C_L(u)\subseteq S$ of labels of the
leaves that are descendants of $u$.  We shall call $N(u)$ the
\emph{rooted subnetwork of $N$ generated by $u$}, and the set of
leaves $C_{L}(u)$ the \emph{cluster} of $u$.

 A \emph{clade} of a phylogenetic network $N$ is a rooted subnetwork
 of $N$ all whose nodes are tree nodes in $N$ (and, in particular, it
 is a rooted tree).

\subsection{Moret-Nakhleh-Warnow-\textsl{et al}'s reduction process}

Let $N=(V,E)$ be a phylogenetic network on a set $S$ of taxa.  A
subset $U$ of internal nodes of $V$ is said to be \emph{convergent}
when it has more than one element, and all nodes in it have exactly
the same cluster.

The removal of convergent sets is the basis of the \emph{reduction
procedure} introduced in \cite{moret.ea:2004}:
\begin{enumerate}
\item[(0)] Replace every clade by a new `symbolic leaf' labeled with the
names of all leaves in it.

\item[(1)] For every maximal convergent set $U$, remove all internal
descendants of its nodes (including the nodes of $U$).

\item[(2)] For every remaining node $x$ that was a parent of a removed
node $v$, add a new arc from $x$ to every (symbolic) leaf in $C_L(v)$.

(The resulting DAG contains no convergent set of nodes, because
this step does not change the clusters of the surviving nodes.)

\item[(3)] Append to every symbolic leaf representing a clade the
corresponding clade, with an arc from the symbolic leaf to the root of
the clade, and remove the label of the symbolic leaf.

\item[(4)] Replace every node  with only one parent and one child by an arc
from its parent to its only child.

(Since the DAG resulting from (2) contains no set of convergent nodes,
it contains no node with only one child.  Therefore the only possible
nodes with only one parent and one child after step (3) are those that
were symbolic leaves with only one parent.  These are the only nodes
that have to be removed in this step.)
\end{enumerate}

The output of this procedure applied to a phylogenetic
network $N$ on $S$ is a (non necessarily rooted) $S$-DAG, called
the \emph{reduced version} of $N$ and denoted by $R(N)$.  A network
$N$ is \emph{reduced} when $R(N)=N$.  It should be noticed that the
only possible convergent sets in $R(N)$ consist of a hybrid node and
its only child (more specifically, the hybrid node corresponding to a
symbolic leaf with more than one parent, and the root of the
corresponding clade) \cite{cardona.ea:07a}.

Two networks $N_{1}$ and $N_{2}$ are said to be
\emph{indistinguishable} when they have isomorphic reduced versions,
that is, when $R(N_{1})\cong R(N_{2})$.  Moret, Nakhleh, Warnow, et al
argue in \cite[p.  19]{moret.ea:2004} that for reconstructible
phylogenetic networks this notion of indistinguishability (isomorphism
after simplification) is more suitable than the existence of an
isomorphism between the original networks.

\section{Nakhleh's `metric'}

Nakhleh defines in \cite{nakhleh:2007} an \emph{equivalence} on the set of nodes
of a pair of $S$-DAGs inductively as follows.

\begin{definition}
Let $N_1=(V_1,E_1)$ and $N_2=(V_2,E_2)$ be $S$-DAGs (not necessarily
different). Two nodes $u\in V_1$ and $v\in V_2$
are \emph{equivalent}, in symbols $u\equiv v$, when:
\begin{itemize}
\item $u$ and $v$ are both leaves labeled with the same taxon, or

\item for some $k\geq 1$, node $u$ has  exactly $k$ children $u_1,\ldots,u_k$,
node $v$ has exactly $k$ children $v_1,\ldots, v_k$, and $u_i\equiv v_i$ for
every $i=1,\ldots,k$.
\end{itemize}
\end{definition}

The following characterization of node equivalence will be useful.

\begin{definition}
Let $N=(V,E)$ be a DAG. The \emph{nested labeling} $\ell(v)$ of the
nodes $v$ of $N$ is defined by induction on $h(v)$ as follows:

\begin{itemize}
\item If $h(v)=0$, that is, if $v$ is a leaf, then $\ell(v)=\{v\}$,
the one-element set consisting of its label.

\item If $h(v)=m>0$, then all its children $v_1,\ldots,v_k$ have
height smaller then $m$, and hence they have been already labeled:
then, $\ell(v)$ is the multiset of their nested labels,
$$
\ell(v)=\{\ell(v_1),\ldots,\ell(v_k)\}.
$$
\end{itemize}
\end{definition}

Notice that the nested label of a node is, in general, a nested
multiset (a multiset of multisets of multisets of\ldots), hence its
name.  Moreover, the height of a node $u$ is the highest level of
nesting of a leaf in $\ell(u)$ minus 1.

\begin{proposition}
Let $N_1=(V_1,E_1)$ and $N_2=(V_2,E_2)$ be DAGs (not necessarily
different) labeled in a set $S$.  For every $u\in V_1$ and $v\in V_2$,
$u\equiv v$ if, and only if, $\ell(u)=\ell(v)$.
\end{proposition}

\begin{proof}
We prove the equivalence by induction on the height of one of the
nodes, say $u$.

If $h(u)=0$, then it is a leaf, and $\ell(u)$ is the one-element set
consisting of its label.  Thus, in this case, $u\equiv v$ if, and only
if, $v$ is the leaf of $N_2$ with the same label as $u$, and
$\ell(u)=\ell(v)$ if, and only if, $v$ is the leaf of  $N_2$ with the
same label as $u$, too.

Consider now the case when $h(u)=m>0$ and assume that the thesis holds
for all nodes $u'\in V_1$ of height smaller than $m$.  Let
$u_1,\ldots, u_k$ be the children of $u$.  Then:
\begin{itemize}
\item $u\equiv v$ if and only if $v$ has  exactly $k$ children and they can be
ordered $v_1,\ldots,v_k$ in such a way that $u_i\equiv v_i$ for every
$i=1,\ldots,k$.

\item $\ell(u)=\ell(v)$ if and only if $v$ has exactly  $k$ children and the
multiset of their nested labels is equal to the multiset of nested
labels of $u_1,\ldots,u_k$, which means that $v$'s children can be
ordered $v_1,\ldots,v_k$ in such a way that $\ell(u_i)=\ell(v_i)$ for
every $i=1,\ldots,k$.
\end{itemize}
Since, by induction, the children of $u$ satisfy the thesis, it is
clear that $u\equiv v$ is equivalent to $\ell(u)=\ell(v)$.
\end{proof}

We shall say that a nested label $\ell(v)$ is \emph{contained} in a
nested label $\ell(u)$, in symbols $\ell(v)\contained \ell(u)$, when $\ell(v)$ is the nested label of a
descendant of $u$.  Notice that if $\ell(v)$ is contained in
$\ell(u)$, then $v$ is equivalent to some
descendant of $u$, but $v$ itself need not be a descendant of $u$:
several instances of this fact can be detected in the networks depicted in
Fig.~\ref{fig:contr1}. Notice moreover that $\ell(v)\in \ell(u)$ if, and only if,
$\ell(v)$ is the nested label of a child of $u$.

Nakhleh defines in \cite{nakhleh:2007} the following dissimilarity
measure.

\begin{definition}
For every $S$-DAG $N$, let $\Upsilon(N)$ be the multiset of equivalence
classes of its nodes (where each equivalence class appears with
multiplicity the number of nodes in it).
\end{definition}

\begin{definition}
For every pair of phylogenetic networks $N_1$ and $N_2$ on the same set $S$ of taxa, let 
$$
m(N_1,N_2)=\frac{1}{2}|\Upsilon(N_1) \bigtriangleup \Upsilon(N_2)|,
$$
where $\bigtriangleup $ denotes the symmetric difference of multisets:
if a class belongs to $\Upsilon(N_1)$ with multiplicity $a$ and to
$\Upsilon(N_2)$ with multiplicity $b$, then it contributes $|a-b|$ to
$| \Upsilon(N_1)\bigtriangleup \Upsilon(N_2)|$.
\end{definition}

Notice that $\Upsilon(N)$ can be also understood as the multiset of
nested labels of the nodes of $N$, each nested label appearing with
multiplicity the number of nodes labeled with it.

\begin{lemma}
\label{prop:pre-nak}
Let $N_1$ and $N_2$ be two $S$-DAGs  such that no one of them contains any pair of equivalent nodes. Then, $m(N_1,N_2)=0$ if, and only if, $N_1\cong N_2$.
\end{lemma}

\begin{proof}
Let $R(N_1)=(V_1,E_1)$ and $R(N_2)=(V_2,E_2)$. If neither $N_1$ nor $N_2$ contain any pair 
of equivalent nodes, then  $\Upsilon(N_1)$ and $\Upsilon(N_2)$ are sets, and the quotient mappings $V_i\to \Upsilon( N_i)$ are bijections, for $i=1,2$.

Now, assume that $|\Upsilon( N_1) \bigtriangleup \Upsilon(N_2)|=0$. Then
$\Upsilon(N_1)=\Upsilon(N_2)$ and hence there exists a well-defined bijection
$\alpha:V_1\to V_2$ that sends each node in $N_1$ to the only node in $N_2$ equivalent to it.
In particular it sends each leaf of $N_1$ to the leaf of $N_2$ with the same label.
To see that $\alpha$ is an isomorphism of graphs, let $(u,v)\in E_1$ be any arc in $N_1$. Since $u\equiv \alpha(u)$, the node $\alpha(u)$ must have a child equivalent to $v$, and since $N_2$ does not contain any pair of equivalent nodes, 
this child is $\alpha(v)$, which implies that $(\alpha(u),\alpha(v))\in E_2$. This shows that $\alpha$ preserves arcs, and a similar argument applied to $\alpha^{-1}:V_2\to V_1$ shows that it also reflects them.
This proves that $\alpha:N_1 \to N_2$ is an isomorphism of $S$-DAGs.

The converse implication is obvious.
\end{proof}

A first consequence of this lemma is the following result, which is essentially Theorem 2 in Nakhleh's paper \cite{nakhleh:2007}.

\begin{proposition}
\label{prop:nak}
Let $R(N_1)$ and $R(N_2)$ be the reduced versions of two phylogenetic
networks on the same set $S$ of taxa.  Then, $m(R(N_1),R(N_2))=0$ if, and only if, $R(N_1)\cong
R(N_2)$.
\end{proposition}

\begin{proof}
The reduced version of a phylogenetic network does not contain any pair of equivalent nodes
  \cite[Obs. 2]{nakhleh:2007}. 
 \end{proof}

\begin{corollary}
Let $N_1$ and $N_2$ be two reduced phylogenetic networks on the same set $S$ of taxa. Then, $m(N_1,N_2)=0$ if, and only if, $N_1\cong N_2$.
\end{corollary}

Another type of phylogenetic networks not containing any pair of equivalent nodes are the \emph{tree-child} phylogenetic networks: phylogenetic networks where every internal node has a child that is a tree node. Tree-child phylogenetic networks were introduced in \cite{cardona.ea:07b}, and a metric and an alignment method for them was proposed, and they have been recently proposed by S. J. Willson as the class where meaningful phylogenetic networks should be
searched~\cite{willson:07}.

\begin{proposition}
A tree-child phylogenetic network does not contain any pair of equivalent nodes. 
\end{proposition}

\begin{proof}
Let $u$ and $v$ be two nodes of  a tree-child phylogenetic network $N$.
If $u\equiv v$, then $C_L(u)=C_L(v)$ and $h(u)=h(v)$. Let now $s$ be a leaf for which there exists a path $u\pathgr s$ with all intermediate nodes and $s$ itself tree nodes (which exists by \cite[Lem.~2]{cardona.ea:07a}).
Then $s\in C_L(u)=C_L(v)$, which implies that there is also a path $v\pathgr s$.
By \cite[Lem.~1]{cardona.ea:07a}, this implies that $u$ and $v$ are connected by a path. If they have moreover the same height, they must be the same node.
\end{proof}

\begin{corollary}
Let $N_1$ and $N_2$ be two tree-child phylogenetic networks on the same set $S$ of taxa. Then, $m(N_1,N_2)=0$ if, and only if, $N_1\cong N_2$.
\end{corollary}

\begin{remark}
\label{rem2}
It is false in general that if two arbitrary phylogenetic networks
$N_1$ and $N_2$ on the same set $S$ of taxa are such that $m(N_1,N_2)=0$, then $N_1\cong N_2$.
For instance, it is easy to check that the networks depicted in
Fig.~\ref{fig:contr1} have the same multisets $\Upsilon$, but they are
not isomorphic.
\end{remark}

\begin{figure}[htb]
\begin{center}
\begin{tikzpicture}[scale=0.9,thick,>=stealth]
\draw(0.5,0) node[tre] (1) {}; \etq 1
\draw(1,0) node[tre] (2) {}; \etq 2
\draw(2.5,0) node[tre] (3) {}; \etq 3
\draw(3,0) node[tre] (4) {}; \etq 4
\draw(0.5,1) node[hyb] (A) {}; 
\draw(1,1) node[hyb] (B) {}; 
\draw(2.5,1) node[hyb] (C) {}; 
\draw(3,1) node[hyb] (D) {}; 
\draw(0,2) node[tre] (a) {}; 
\draw(1.25,2) node[tre] (b) {}; 
\draw(2.25,2) node[tre] (c) {}; 
\draw(3.5,2) node[tre] (d) {}; 
\draw(0,3) node[hyb] (E) {}; 
\draw(1.25,3) node[hyb] (F) {}; 
\draw(2.25,3) node[hyb] (G) {}; 
\draw(3.5,3) node[hyb] (H) {}; 
\draw(0,4.5) node[tre] (e) {}; 
\draw(1.25,4.5) node[tre] (f) {}; 
\draw(2.25,4.5) node[tre] (g) {}; 
\draw(3.5,4.5) node[tre] (h) {}; 
\draw(0.65,5.5) node[tre] (i) {}; 
\draw(2.87,5.5) node[tre] (j) {}; 
\draw(1.75,6.5) node[tre] (r) {}; 
 \draw[->](r)--(i);
 \draw[->](r)--(j);
 \draw[->](i)--(e);
    \draw[->](i)--(f);
    \draw[->](j)--(g);
    \draw[->](j)--(h);
    \draw[->](e)--(E);
    \draw[->](f)--(F);
    \draw[->](g)--(G);
    \draw[->](h)--(H);
    \draw[->](e)--(G);
    \draw[->](f)--(H);
    \draw[->](g)--(E);
    \draw[->](h)--(F);
    \draw[->](E)--(a);
    \draw[->](F)--(b);
    \draw[->](G)--(c);
    \draw[->](H)--(d);
    \draw[->](a)--(A);
    \draw[->](a)--(B);
    \draw[->](b)--(A);
    \draw[->](b)--(B);
    \draw[->](c)--(C);
    \draw[->](c)--(D);
    \draw[->](d)--(C);
    \draw[->](d)--(D);
    \draw[->](A)--(1);
    \draw[->](B)--(2);
    \draw[->](C)--(3);
    \draw[->](D)--(4);
  \end{tikzpicture}
  \qquad\qquad\qquad
  \begin{tikzpicture}[scale=0.9,thick,>=stealth]
\draw(0.5,0) node[tre] (1) {}; \etq 1
\draw(1,0) node[tre] (2) {}; \etq 2
\draw(2.5,0) node[tre] (3) {}; \etq 3
\draw(3,0) node[tre] (4) {}; \etq 4
\draw(0.5,1) node[hyb] (A) {}; 
\draw(1,1) node[hyb] (B) {}; 
\draw(2.5,1) node[hyb] (C) {}; 
\draw(3,1) node[hyb] (D) {}; 
\draw(0,2) node[tre] (a) {}; 
\draw(1.25,2) node[tre] (b) {}; 
\draw(2.25,2) node[tre] (c) {}; 
\draw(3.5,2) node[tre] (d) {}; 
\draw(0,3) node[hyb] (E) {}; 
\draw(1.25,3) node[hyb] (F) {}; 
\draw(2.25,3) node[hyb] (G) {}; 
\draw(3.5,3) node[hyb] (H) {}; 
\draw(0,4.5) node[tre] (e) {}; 
\draw(1.25,4.5) node[tre] (f) {}; 
\draw(2.25,4.5) node[tre] (g) {}; 
\draw(3.5,4.5) node[tre] (h) {}; 
\draw(0.65,5.5) node[tre] (i) {}; 
\draw(2.87,5.5) node[tre] (j) {}; 
\draw(1.75,6.5) node[tre] (r) {}; 
 \draw[->](r)--(i);
 \draw[->](r)--(j);
 \draw[->](i)--(e);
    \draw[->](i)--(f);
    \draw[->](j)--(g);
    \draw[->](j)--(h);
    \draw[->](e)--(E);
    \draw[->](f)--(E);
    \draw[->](g)--(F);
    \draw[->](h)--(F);
    \draw[->](e)--(G);
    \draw[->](f)--(G);
    \draw[->](g)--(H);
    \draw[->](h)--(H);
    \draw[->](E)--(a);
    \draw[->](F)--(b);
    \draw[->](G)--(c);
    \draw[->](H)--(d);
    \draw[->](a)--(A);
    \draw[->](a)--(B);
    \draw[->](b)--(A);
    \draw[->](b)--(B);
    \draw[->](c)--(C);
    \draw[->](c)--(D);
    \draw[->](d)--(C);
    \draw[->](d)--(D);
    \draw[->](A)--(1);
    \draw[->](B)--(2);
    \draw[->](C)--(3);
    \draw[->](D)--(4);
  \end{tikzpicture}
\end{center}
\caption{\label{fig:contr1} These phylogenetic networks have the same
multisets of equivalence classes of nodes, but they are not
isomorphic}
\end{figure}

Now, it turns out that this metric $m$  separates
networks that are distinguishable up to reduction.  We would like to
recall here that this was the (unaccomplished
\cite{cardona.ea:07a}) goal of the error metric defined in
\cite{moret.ea:2004}.

\begin{theorem}
Let $N_1$ and $N_2$ be two phylogenetic networks on the set $S$ of
taxa. If
$m(N_1,N_2)=0$, then $N_1$ and $N_2$ are indistinguishable.
\end{theorem}

\begin{proof}
In this proof, we shall take $\Upsilon(N)$ as the multiset of nested
labels of a network $N$.  Let $N_1=(V_1,E_1)$ and $N_2=(V_2,E_2)$ be two phylogenetic
networks such that $\Upsilon(N_1)=\Upsilon(N_2)$.  We shall prove that
the reduction process of both networks modifies exactly in the same
way their multisets of nested labels, and thus the reduced
versions $R(N_1)$ and $R(N_2)$ also have the same multisets of nested
labels.  Then, by Proposition \ref{prop:nak}, the latter are
isomorphic.

To begin with, notice that two nodes are convergent when the set of
$S$-labels appearing in their nested labels are the same (without
taking into account nesting levels or multiplicities).  In particular,
$N_1$ and $N_2$ have the same sets of nested labels of convergent
nodes.

Step (0) in the reduction process consists of replacing every clade by
a symbolic leaf.  This corresponds to remove the
nested labels of the nodes belonging to clades (except their roots)
and to replace, in all remaining nested labels, each nested label of a
root of a clade by the label of the corresponding symbolic leaf. We 
must prove now that we can decide from the multisets of nested labels alone which 
are the nested labels of nodes of clades and of roots of clades. 

Since  the clades of a phylogenetic network are subtrees, a node
belonging to a clade is only equivalent to itself (if $v$ is a node of a clade and $v\equiv u$, then $C_L(u)=C_L(v)$, but in this case, since $v$ is the least common ancestor of $C_L(v)$ in the clade it belongs,
$v$ must be a descendant of $u$, and since $u$ and $v$ have the same height ---because they are equivalent--- they must be the same node). In particular, a node of a clade does not share its nested label with any other node.

Then, 
the nested labels of nodes $v\in V_i$  belonging to some clade of $N_i$ ($i=1,2$) are characterized by
the following two properties: 
$\ell(v)$ and each one of  the nested labels contained in it appear with multiplicity 1 in $\Upsilon(N_1)=\Upsilon(N_2)$ (and in particular $v$ and its descendants are characterized by their nested labels);
and  $\ell(v)$ and each one of  the nested labels contained in it belong at most to one nested label
(this means that $v$ and its descendants are tree nodes, and in particular that the rooted subnetwork generated by $v$ is a tree consisting only of tree nodes from $N_i$).
And therefore the roots of clades of $N_i$ are the nodes $v$ with nested label $\ell(v)$ maximal with these properties, and the nodes of
the clade rooted at $v$ are those nodes with nested labels contained in
$\ell(v)$.  This shows that the nested
labels of roots of clades and the nested labels of nodes belonging to
clades in $N_{1}$ are the same as in $N_{2}$.

So, we remove the same nested labels in $N_{1}$ and $N_{2}$ and we 
replace the same nested labels by symbolic leaves. As a consequence, 
the networks resulting after this step have the same nested labels.

In step (1), all nodes that are convergent with some other node are
removed, and all nodes other than symbolic leaves that are
descendant of some removed node are also removed.  So, in this step we
remove the nested labels of convergent nodes, and the nested labels
other than singletons that are contained in some nested label of
convergent node (notice that if $\ell(v)$ is not a singleton and it is
contained in $\ell(u)$ and $u$ is convergent, then either $v$ is a
descendant of $u$, and then it has to be removed, or it is equivalent
to a descendant of $u$, and then it forms a convergent set with this
descendant and it has to be removed, too).  This shows that the nested
labels of the nodes removed in both networks are the same, and hence that  the
nested labels of the nodes that remain in both networks are also the
same.

In step (2), the paths from the remaining nodes to the labels are
restored.  It means to replace in each
remaining nested label $\ell(x)$, each maximal  
nested label $\ell(v)\contained \ell(x)$ of a removed node $v$  by the
singletons $\{s_1\},\{s_2\},\ldots,\{s_p\}$ of the symbolic leaves
appearing in $\ell(v)$.  Again, this operation only depends on the
nested labels, and therefore after this step the resulting DAGs have
the same multisets of nested labels.

In step (3), clades are restored. This is simply done by replacing in
the nested labels each symbolic leaf $s$ by the nested label of the
root of the clade it replaced, between brackets (because we append it
to the node corresponding to the symbolic leaf).  Since the same
clades were removed in both networks and replaced by the same symbolic leaves, after this step the resulting
DAGs still have the same multisets of nested labels.

Finally, in step (4), the nodes with only one parent and only one
child are removed.  This corresponds to remove nested labels of the
form $\{\{\ldots\}\}$  that are children of only one parent (that
is, that belong to only one nested
label), and hence the same nested labels are removed in both DAGs.

So, at the end of this procedure, the resulting DAGs $R(N_1)$ and
$R(N_2)$ have the same multisets of nested labels.  By Proposition
\ref{prop:nak}, this implies that $R(N_1)$ and $R(N_2)$ are
isomorphic.
\end{proof}

The converse implication is, of course false: since the reduction
process may remove parts with different topologies that yield
differences in the multisets of equivalence classes, two phylogenetic
networks with isomorphic reduced versions may have different multisets
of equivalence classes.

The value $m(N_1,N_2)$ can be computed in time polynomial in
the sizes of the networks $N_1,N_2$  by performing a simultaneous bottom-up traversal of the two
networks \cite{valiente:2002,valiente:2007}

\section{A metric for arbitrary phylogenetic networks}

If instead of the equivalence classes of the
nodes (or, equivalently, their nested labels) we consider the whole
rooted subnetworks generated by the nodes, we can define a true
distance on the whole class of all phylogenetic networks.

\begin{remark}
It is clear that if $u$ and $v$ are two nodes of two phylogenetic
networks $N_{1}$ and $N_{2}$, respectively  (it can happen that
$N_{1}=N_{2}$), such that the rooted subnetworks $N_{1}(u)$ and
$N_{2}(v)$ generated by them are isomorphic, then $u\equiv v$
(because the equivalence can be computed within these rooted
subnetworks). But the converse implication is false:
node equivalence in phylogenetic networks does not imply 
isomorphism of the rooted subnetworks.  Consider for
instance the non-isomorphic phylogenetic networks depicted in
Fig.~\ref{fig:contr1}: it is easy to check that their roots are
equivalent.
\end{remark}

\begin{definition}
For every $S$-DAG $N$, let $\Sigma(N)$ be the multiset of isomorphism
classes of the rooted subnetworks generated by its nodes.
\end{definition}

\begin{definition}
For every pair of phylogenetic networks $N_1$ and $N_2$  on the same set $S$ of
taxa, let 
$$
\sigma(N_1,N_2)=\frac{1}{2}|\Sigma(N_1) \bigtriangleup \Sigma(N_2)|,
$$
where $\bigtriangleup $ denotes the symmetric difference of multisets.
\end{definition}

\begin{theorem}
Let $N_1$ and $N_2$ be two phylogenetic networks on the same set $S$ of
taxa.  Then, $\sigma(N_1,N_2)=0$ if, and only if, $N_1\cong N_2$.
\end{theorem}

\begin{proof}
Assume that $\sigma(N_1,N_2)=0$, that is, $\Sigma(N_1)=\Sigma(N_2)$.
Since each $N_i$ is its rooted subnetwork  generated by its
root, we conclude that $N_1$ contains a rooted subnetwork isomorphic
to $N_2$ and $N_2$ contains a rooted subnetwork isomorphic to $N_1$.
The only possibility is then that $N_1$ and $N_2$ are isomorphic
(otherwise, $N_1$ would contain a rooted subnetwork isomorphic to it
and strictly contained in it, something that in finite graphs is
impossible).

The converse implication is obvious.
\end{proof}

\begin{corollary}
The mapping $\sigma$ is a metric on the class of all phylogenetic
networks on the set $S$ of taxa, that is, it satisfies the following
properties: for every phylogenetic networks $N_1,N_2,N_3$ on the set 
$S$,
\begin{enumerate}[(a)]
\item \emph{Non-negativity}: $\sigma(N_1,N_2)\geq 0$
\item \emph{Separation}: $\sigma(N_1,N_2)=0$ if and only if $N_1\cong N_2'$
\item \emph{Symmetry}: $\sigma(N_1,N_2)=\sigma(N_2,N_1)$
\item \emph{Triangle inequality}: $\sigma(N_1,N_3)\leq \sigma(N_1,N_2)+\sigma(N_2,N_3)$
\end{enumerate}
\end{corollary}

\begin{proof}
Properties (a) and (d) are straightforward, property (b) is a
consequence of the last theorem, and property (d) is a consequence of
the triangle inequality of the symmetric difference of multisets.
\end{proof}

The computation of $\sigma$ has at least the same complexity as the $S$-DAG isomorphism problem (because the latter can be decided using $\sigma$), and isomorphism of general DAGs can be reduced to  $S$-DAG isomorphism. Therefore, the problem of deciding whether $\sigma$ can be computed in polynomial time for arbitrary phylogenetic networks remains open. But if we bound the in and out-degree of the nodes, the $S$-DAG isomorphism problem is in P, and therefore $\sigma$ can be computed in polynomial time by performing a simultaneous bottom-up traversal of the two
networks.

\section{Conclusion}
In this paper we have complemented Luay Nakhleh's latest proposal of a
metric $m$ for phylogenetic networks by (a) showing that $m$ separates
distinguishable networks,  and (b) proposing a
modification of its definition that provides a true metric $\sigma$ on
the class of all phylogenetic networks. 
When both distances $m$ and $\sigma$   are applied to phylogenetic trees, they both
yield half the symmetric differences of the sets of (isomorphism
classes of) subtrees.

The measure $m$ can be computed in time  polynomial in the size  of the networks, but since $\sigma$ can be used to decide the isomorphism problem for $S$-DAGs,
we are lead to conjecture that it cannot be computed in polynomial time (as any other dissimilarity measure for phylogenetic networks satisfying the separation property). Any way, $\sigma$ can also computed in polynomial time on subspaces of phylogenetic trees with bounded in and out-degree.

Given a set $S$ of $n\geq 2$ labels, there exists no upper bound for the values of
$\sigma(N_{1},N_{2})$ and $m(N_{1},N_{2})$,
as there exist arbitrarily large phylogenetic networks with $n$ 
leaves and no internal node of any one of them equivalent to an 
internal node of the other one.


\end{document}